\documentclass[letterpaper, 10 pt, conference]{ieeeconf}  

\IEEEoverridecommandlockouts                             
\usepackage{amsfonts,amsmath}
\usepackage{amssymb, amsmath, graphicx, epsfig, cite,subfigure,color}

 \newtheorem{theorem}{Theorem}[section]
 \newtheorem{corollary}[theorem]{Corollary}
 \newtheorem{lemma}[theorem]{Lemma}

 \newtheorem{definition}[theorem]{Definition}
 
 \DeclareMathOperator*{\argmin}{arg\,min}

\title{\LARGE \bf
Capturing Aggregate Flexibility in Demand Response
}

\author{Mahnoosh Alizadeh, Anna Scaglione,  Andrea Goldsmith, George Kesidis
\thanks{Supported by US DOE under CERTS load as a resource program. }
\thanks{Mahnoosh Alizadeh and Anna Scaglione are at UC Davis. Andrea Goldsmith is at Stanford University. Goerge Kesidis is at Pennsylvania State University.
        Email: malizadeh@ucdavis.edu}%
}

\begin{document}

\maketitle
\thispagestyle{empty}
\pagestyle{empty}

\begin{abstract}
Flexibility in electric power consumption  can be leveraged by Demand Response (DR) programs. 
The goal of this paper is to systematically capture the inherent aggregate flexibility of a population of 
appliances. We do so by clustering individual loads based on their characteristics and service constraints. 
We highlight the challenges associated with learning the customer response to economic incentives while applying demand side management
to heterogeneous appliances. We also develop a framework to quantify customer privacy in direct load scheduling programs.
\end{abstract}

\section{INTRODUCTION}
It is widely understood that certain categories of electric loads are flexible. The challenge is to design control schemes and economic incentives to harness the flexibility of electric loads with very heterogeneous characteristics. A specific design for these control and economic aspects is what constitutes a Demand Response (DR) scheme. A large and heterogeneous populations of electric loads is challenging to model because the  demand at each time unit is stochastic, and its collective response to control signals depends on the response of individual members of the appliance population.

\subsubsection*{Contribution}
The response of an individual flexible load to control signals depends on the inherent set of loads that the customer can choose from, which we refer to as {\it load plasticity}. These responses add up to define the aggregate load response. The goal of this paper is to mathematically capture the {\it inherent} plasticity of aggregate load using the principles of quantization, and design demand response architectures that leverage these models. To tackle  appliance heterogeneity,  we use a clustering approach that quantizes the parameters used to describe the plasticity of each individual appliance in the population. We illustrate the model fundamentals considering loads that are canonical batteries, and as practical cases we only discuss deferrable loads for brevity. 
We then determine  how load plasticity can be exploited by two prominent types of DR, namely direct load scheduling and dynamic retail pricing. 
We also use mutual information bounds to quantify the leakage of information about private consumption profiles. 

\subsubsection*{Related work}
The growing literature on DR includes a wide range of load models for describing heterogeneous appliance characteristics.  Most commonly, previous works focus on the flexibility available to the DR scheme of interest, and not the underlying inherent flexibility of load that is the subject of this work. 

One approach for appliance modeling is to preserve all details about the associated states and constraints, which can entail great complexity for large populations. The electric power consumption properties are either idealized (often as a battery) or described realistically via hybrid dynamical systems equations. Examples of works applying market analyses using detailed models for individual appliances are \cite{galus,6084772,6419868,6507354,6145671}.  The complexity of detailed models is manageable when the population and associated number of decisions is small or when the policy is simple. For example, Home Energy Management Systems protocols can use detailed models to respond optimally to prices, because the number of decisions is limited \cite{mohsenian2010autonomous,kefayati2010efficient}.   Many papers considered large populations of Electric Vehicle (EV) \cite{kesidis,chen2012large,subramanian2012real,ramraja,6471273} and explored the optimality of simple policies such as  Least Laxity First (LLF).

Another series of papers propose to capture the total flexibility of deferrable loads for planning and market interactions as a tank that needs to be fully charged by a deadline \cite{homer, 6426102,ortega}. By discarding all specific characteristics of individual appliances except for their total energy consumption, the model has minimal computational cost  for planning, but could lead to  large errors as we see in Sec. \ref{sec.numerical}. 

We use a different approach. Specifically we use the principles of quantization to trade off accuracy for the order of the model. 
The concept is not completely new. In particular, the idea that aggregate load of appliances can be broken up into demands originating from various homogeneous groups of devices was first proposed for thermostatically controlled loads in \cite{1104071,4046405}.  In these works, heterogeneous appliances are subject to
the same control, within a load management program; perturbation analysis is used to provide an approximate description of aggregate load flexibility \cite{1104071}. This  load grouping approach was later used for other applications of load modeling, such as electric vehicle charge scheduling, e.g., \cite{6407491,smartgridcomm}, or to develop Markovian models for residential meter data \cite{ardakanian2011markovian}. Here we provide a unified view of load clustering that can capture the inherent flexibility of large populations and be used for direct scheduling.  Our model captures much higher levels of controllability in electric load, consistent with current practices in DR programs. 

{\bf Notation}: Continuous  variables are in roman font $\rm x(t)$  and discrete variables are in italic $x(t)$. Boldface
is used for vectors. Finite differences are $\partial x(t)=x(t+1)-x(t)$. The unit step is $u(t)$, and   the Kronecker delta function is $\delta(t)$.
The symbol $\star$ denotes discrete time convolution. 
The expected value of a random variable $x$ is denoted by ${\mathbb E}\{x\}$ .

\section{Modeling Aggregate Demand}\label{sec.modeling}\label{sec.loadmodel}
In this section we develop the framework to  trade-off model accuracy and complexity. 
We refer to the underlying potential of load to be modified by control actions as {\it load plasticity}. Mathematically, if an appliance indexed by $i$ becomes available for load control at time ${\rm \tau}_i$, the load plasticity is a set ${\mathcal L}_i({\rm t})$ of load profiles ${\rm L}_i({\rm t})$  that can satisfy the service requirements specified by the end-user at times ${\rm t}>{\rm \tau}_i$.

\subsection{The Canonical Battery Model}  The simplest type of load {\it plasticity} is that offered by the canonical battery category, which we now define. We assume that a canonical battery indexed by $i$ is plugged in at a  time $\tau_i$. We define an arrival process: 
\begin{equation}a_i(t)=u(t-\tau_i),\label{arrivaldef}\end{equation}
that marks the time when the battery is connected to the grid and first available to be controlled. 
It has an initial charge ${\rm S}_i$, a deadline to fully charge ${\rm \chi}_i$, and a capacity ${\rm E}_i$. Also, canonical batteries have no rate limitations and can move from any state to any other state instantaneously. Denoting the current state of charge of the battery by $\rm x_i(t)$, the plasticity of the canonical battery is the set
  \begin{align}
  \rm
{\mathcal L}_i(\rm t)=\big\{\rm L_i(\rm t) |&  {\rm L}_i(\rm t) = \dot {\rm x}_i(\rm t), {\rm x}_i(\tau_i) = {\rm S}_i, {\rm x}_i({\rm \chi}_i) = {\rm E}_i,  \nonumber\\&0 \leq {\rm x}_i(t)\leq {\rm E}_i, {\rm t}\geq {\rm \tau_i}\big\}.\label{lossless-battery}
\end{align}

This model is analog and continuous. Quantizing continuous values and signals is the basic principle of compression and source coding. We advocate using these principles to provide a class of models that is amenable and tractable for demand response.

We first quantize time ${\rm t}$ into equally-distanced discrete epochs $t$ separated by $\delta T$.  Second, we quantize all energy variables $({\rm x}_i(t), {\rm E}_i, {\rm S}_i)$ using a uniform quantizer with step $\delta x$.  Both $\delta T$ and $\delta x$ are equal to 1 for brevity of notation. Thus, the discrete version of the load plasticity for a canonical battery is:
  \begin{eqnarray}\nonumber
{\mathcal L}_i(t)&=&\{L_i(t) |  L_i(t) = \partial x_i(t), x_i(\tau_i) = S_i, x_i(\chi_i) = E_i, \\
				& &x_i(t) \in \{0,1,\ldots, E_i\}, \tau_i \leq t \leq \chi_i \}\label{lossless-battery-discrete}.
\end{eqnarray}
Clearly, there is no flexibility offered before time $\tau_i$ and after time $\chi_i$, i.e.,
\begin{equation}
{\mathcal L}_i(t)=\{0\} ,~~ \forall t<\tau_i, \forall t>\chi_i.
\end{equation}

The appliance type $v$ (canonical battery here), and the following parameters specify  $ {\mathcal L}_i(t)$ uniquely:
$$\boldsymbol{\vartheta}_ i = (\tau_i, S_i, \chi_i, E_i) \in \mathcal{T}^v.$$

\subsection{Population Model for Canonical Batteries}\label{idealbat.pop}
The difference between modeling a single canonical battery and a population of them ($i \in {\cal P}^v$ ) is in capturing the non-homogeneous individual parameters $\boldsymbol{\vartheta}_ i = (\tau_i, S_i, \chi_i, E_i)$ in the description of the aggregate load plasticity. After quantizing the parameters we cluster together all batteries that fall in the same quantization bin. We denote the plasticity of each member of these groups as ${\mathcal L}^v_{\boldsymbol{\vartheta}}(t)$.

We define the following operations on plasticities $\mathcal{A}_1$, $\mathcal{A}_2$:
\begin{align}&\mathcal{A}_1 + \mathcal{A}_2 = \left\{x| \exists (x_1,x_2) \in \mathcal{A}_1\times \mathcal{A}_2,  \mbox{such that} ~ x = x_1 + x_2  \right\}\nonumber\\
&n \mathcal{A}_1 = \left\{x| \exists x_i \in \mathcal{A}_1, i =1,\ldots, n,  \mbox{such that} ~ x = \sum_{i=1}^n x_i  \right\},\nonumber\end{align}
where $n \in \mathbb{N}$ and $0 \mathcal{A}_1 \equiv \{0\}$.

We denote as $n^v_{\boldsymbol{\vartheta}}$ the total number of batteries with $\boldsymbol{\vartheta} \in \mathcal{T}^v$.
Then, the aggregate plasticity of the population ${\cal P}^v$  is:
\begin{eqnarray}
{\mathcal L}^v(t)&=& \sum_{\boldsymbol{\vartheta}\in \mathcal{T}^v}n^v_{\boldsymbol{\vartheta}}{\mathcal L}^v_{\boldsymbol{\vartheta}}(t).
\end{eqnarray}

Next, we wish to simplify the description of plasticity ${\mathcal L}^v(t)$ by decreasing the number of variables that define it. We do so by exploiting similarities in the sets ${\mathcal L}^v_{\boldsymbol{\vartheta}}(t)$. 

Generally, the parameters that describe an individual appliance load can be divided in two subsets: one set of parameters, denoted by $\boldsymbol{\kappa}_i$, describe the initial state of  control variables. For example, $S_i$ and $\tau_i$ here define the initial arrival event of each battery. Changing these quantities affects the load plasticity of an appliance only in a transient fashion, but does not affect the underlying structure of the plasticity. Appliances that only differ in terms of these initial parameters can be bundled together in a single population model for the load plasticity, as we will see next. Parameters in $\boldsymbol{\vartheta}_i$ not included in $\boldsymbol{\kappa}_i$ are denoted by $\boldsymbol{\theta}_i$.

\subsubsection{Clustering Batteries with Homogeneous $\boldsymbol{\theta}_i = (\chi_i, E_i)$}
Here we simplify the aggregate plasticity of batteries that have the same $(\chi, E)$, but have non-homogeneous $S_i$ and $\tau_i$. 
We denote as $a_x(t)$ the total number of batteries that arrive in the system with an initial state of charge $S_i$ equal to $x$ at or before time $t$. The value of $a_x(t)$ can be written in terms of individual arrival processes $a_i(t)$ as:
\begin{equation}
 a_x(t)=\sum_{i\in {\cal P}^v}\delta(S_i-x)  a_i(t).
\end{equation}
We refer to $a_x(t)$ as the {\it arrival process for state $x$}. We also denote the total number of batteries in state $x$ at time $t$ as $n_x(t)$, where
\begin{equation}\label{numberineachstate}
 n_x(t)=\sum_{i\in {\cal P}^v}\delta(x_i(t)-x)a_i(t),~x=1,\ldots,E.
\end{equation}

Using \eqref{lossless-battery-discrete} and \eqref{arrivaldef},  the total load of the batteries is:
\begin{equation}\label{pop}
L(t)=\sum_{i\in {\cal P}^v}L_i(t)=\sum_{i\in {\cal P}^v}\partial x_i(t) a_i(t).
\end{equation}
Next, we directly tie the evolution of $n_x(t)$ and $a_x(t)$ to the total load $L(t)$, removing all dependence on $x_i(t)$ and $a_i(t)$.

\begin{lemma}\label{nxlemma}{\it 
The following relationship holds between $n_x(t)$ and the load $L(t)$:
\begin{align}\label{occur}
L(t) & = \sum_{x=0}^{E}\left[\left(\sum_{x'=x}^{E} \partial n_{x'}(t) \right) - (x+1)\partial a_{x}(t)\right].
\end{align}
}
\end{lemma}
Proof omitted for brevity.

Lemma \ref{nxlemma} shows that to model or control the load $L(t)$,   the Aggregator needs to keep track of $a_x(t)$, and track/control the evolution of $n_x(t)$. Clearly, $n_x(t)$ cannot take arbitrary values. Specifically, how appliances are allowed to move from one state to another constrains the evolution of $n_x(t)$. For canonical batteries, the set of possible control actions is easy to specify, as each battery can move from one state $x$ to any other $x'$ in just one time step by consuming or drawing the charge $x'- x$.

\begin{definition}{\it  We denote the total number of batteries that go from state $x$ to state $x'$ at time step $t$ as $d_{x,x'}(t)$. We refer to $d_{x,x'}(t)$  as the {\it switch process from state $x$ to $x'$}. By definition,  $d_{x,x}(t) = 0, ~\forall t,x$, and $d_{x,x'}(0) = 0,~\forall x,x'$.}
\end{definition}
\begin{corollary}{\it
The occupancy $n_x(t)$ and aggregate load $L(t)$ in terms of $d_{x,x'}(t)$ are: }
\begin{align}\label{nxt2}
n_{x}(t+1)&=  a_x(t+1) + \sum_{x' = 0}^E [d_{x',x}(t)- d_{x,x'}(t)]
\\\label{occur3}
L(t) &= \sum_{x=0}^E \sum_{x'=0}^E (x'-x)\partial d_{x,x'}(t)
\end{align}
\end{corollary}
\begin{proof}
The occupancy at time $t+1$ should include the previous occupancy plus new arrivals from other states or from outside, minus the population that exits the state:
\begin{align}\label{nxt}
\partial n_{x}(t)&=\partial a_x(t) + \sum_{x' = 0}^E \partial [d_{x',x}(t)-d_{x,x'}(t)]
\end{align}
which leads to \eqref{nxt2} if summed over time. If we substitute the value of $\partial n_x(t)$ in \eqref{nxt} into \eqref{occur}, we get \eqref{occur3}.\end{proof}
Thus, the load plasticity can be presented in terms of the $d_{x,x'}(t)$'s, under appropriate constraints:
  \begin{align}\nonumber
{\mathcal L}^v(t) \!=& \! \! \bigg\{L(t) |  L(t) \!=\!\! \sum_{x=0}^E \sum_{x'=0}^E (x'-\!x)\partial d_{x,x'}(t), \partial d_{x,x'}(t) \!\in \!\mathbb{Z}^+, \\ &\sum_{x' = 1}^{E}\partial d_{x,x'}(t) \leq n_x(t), n_x(\chi) = 0, \forall x < E \bigg\}\label{occur4}
\end{align}
where $n_x(t)$ is given by \eqref{nxt2}. The second constraint ensures that only batteries present in state $x$ can be moved from $x$ to any $x'$. The third ensures that the deadline is met.

The reader should note the simplicity of the  load population model in \eqref{occur4}, since it only contains linear  equations and constraints. Also, 
for large populations $d_{x,x'}(t)$'s can be approximated with a real number. The model is scalable since load plasticity is in an affine space and requires tracking $(a_x(t),n_x(t))$, for the limited number of quantized states $E$ considered, and deciding $E^2$ values $d_{x,x'}(t)$, no matter how large the population. 

\subsubsection{Bundling Batteries with Non-homogeneous $\boldsymbol{\vartheta}_i$}\label{cluster.sec}
Remember that we divided appliance parameters into two parts, i.e., $\boldsymbol{\vartheta}_i = (\boldsymbol{\theta}_i, \boldsymbol{\kappa}_i) $.
Unlike the parameters  in $\boldsymbol{\kappa}_i = (\tau_i,S_i)$, the parameters in $\boldsymbol{\theta}_i$ change the set of decision variables $d_{x,x'}(t) $ and constraints that describe the load plasticity. Thus, we cannot bundle batteries with non-homogeneous $\boldsymbol{\vartheta}_i$ all together.
We bundle requests with similar constraints captured through $\boldsymbol{\theta}_i $  in {\it clusters}  indexed by $q=1, \ldots, Q^v$:
\begin{equation}
\boldsymbol{\theta}_i \xrightarrow[\mathcal{Q}]{\mbox{Quantize}} \boldsymbol{\theta}_q \xrightarrow[\mathcal{I}] {\mbox{Cluster index}} q.
\end{equation} 
The level of quantization error can be controlled by modifying $Q$ or $\boldsymbol{\theta}_q$'s, and is the knob that controls the complexity and accuracy of the aggregate model. 

We use a superscript $q$ to refer to any previously defined set of decision variables $d^q_{x,x'}(t)$ for cluster $q$. Thus, generalizing \eqref{occur4} to the non-homogeneous case,
  \begin{eqnarray}\nonumber
{\mathcal L}^v(t) = \Bigg\{ L(t) |  L(t) = \sum_{q=1}^{Q^v}\sum_{x=0}^{E^q} \sum_{x'=0}^{E^q} (x'-x)\partial d^q_{x,x'}(t)~~\nonumber 
& & \\ \partial d_{x,x'}^q(t) \in \mathbb{Z}^+, \forall x,x' \in \{0,1,\ldots, E^q\}~~& &  \nonumber \\ \sum_{x' = 1}^{E^q}\partial d_{x,x'}^q(t)\leq n^q_x(t), n_x(\chi^q) = 0, \forall x < E^q \Bigg\}&& \label{occur5} 
\end{eqnarray}
with
$n_{x}^q(t)=  a^q_x(t) + \sum_{x' = 0}^{E^q} [d^q_{x',x}(t-1)- d^q_{x,x'}(t-1)]$.

This gives us a hybrid stochastic model for the load plasticity of a population of canonical batteries. Next, in Section \ref{pop.gen}, we will build on this model to present the aggregate load plasticity of more realistic deferrable loads.


\section{Population Models for Realistic Loads}\label{pop.gen}
Here we build on the canonical battery population model to tackle other categories of real appliances. For brevity only {\it deferrable} loads are considered. Appliances in each category $v$ are bundled together in a single population model, whose plasticity set is $\mathcal{L}^v(t)$. There are $Q^v$ clusters for each category $v$.
In general, an Aggregator can serve $V$ different categories of loads. Given the population load plasticity of each category ${\mathcal L}^v(t)$, the load plasticity of the total demand served by an Aggregator as the set:
\begin{align}\label{totplas}
{\mathcal L}(t)= {\mathcal L}^I(t)+ \sum_{v=1}^V {\mathcal L}^v(t)  
\end{align}
where ${\mathcal L}^I(t)$ denotes the plasticity of the inflexible demand served by the Aggregator, which has only one set member:
\begin{align}\label{totplas}
{\mathcal L}^I(t)=\{ L^{I}(t)\}
\end{align}

\subsection{Rate-Constrained Instantaneous Consumption (RIC)} \label{cat1}
This class is a non-homogeneous population of batteries, each characterized by the vector $(\tau_i, X_i , E_i, \chi_i, \rho_i, G_i)$, where $\chi_i$   denotes the deadline for battery $i$ to receive at least $\rho_i$ percent of full charge $E_i$, and $G_i$ denotes the maximum rate at which battery $i$ can be charged/discharged. 
This model could be used, for example, for Electric Vehicle charging and Vehicle to the Grid (V2G) applications. 

The $i$th load plasticity is:
  \begin{eqnarray}\nonumber
{\mathcal L}_i(t)&=&\{L_i(t) |  L_i(t) = \partial x_i(t), x_i(\tau_i) = S_i, \\
				& &x_i(t) \in \{0,1,\ldots, E_i\}, x_i(\chi_i) \geq \rho_i E_i, \nonumber \\
& & -G_i \leq \partial x_i(t) \leq G_i, t\geq \tau_i \}  \label{lossless-battery-discrete2}
\end{eqnarray}
Writing the population plasticity is straightforward following the steps taken for canonical batteries. We cluster batteries using $\boldsymbol{\theta}_i = (E_i, \chi_i, \rho_i, G_i)$, and denote the parameters associated with cluster $q$ as $(E^q, \chi^q, \rho^q, G^q)$. 
In each cluster $q$:
\begin{enumerate}
\item Contrary to the canonical battery case, due to rate limitations, at each time step, a battery can move by a maximum of $G^q$ states, i.e. an appliance can draw or inject power at a rate $|\partial x(t)|\leq G^q$. Thus, from state $x$, an appliance can move only to one of the following states:
\begin{equation}
{\mathcal S}_x =\{x - G^q, x - G^q+1, \ldots, x + G^q\};
\end{equation}
\item All appliances in cluster $q$ should be in a state $x \geq \rho^q E^q$ by $\chi^q$. Equivalently $n_x^q(\chi^q) = 0, \forall x < \rho^q E^q$.
\end{enumerate}

Consequently, the load plasticity of the population is:
  \begin{align} \label{occur5} 
{\mathcal L}^v(t) & \!= \!\!\Big\{L(t) | \! L(t) \!= \!\!\sum_{q=1}^{Q^v}\!\sum_{x=0}^{E^q}\!\sum_{x'\in \mathcal S_x}\!\!(x'\!-\!x)\partial d^q_{x,x'}(t), \partial d_{x,x'}^q(t)\! \in \!\mathbb{Z}^+ \nonumber
\\ &\!\!\sum_{x'\in \mathcal S_x}\!\!\partial d_{x,x'}^q(t)\leq n^q_x(t),\!  n_x^q(\chi^q) \!= \!0, \forall x\! < \!\rho^q E^q\Big\}, 
\end{align}
with
$n_{x}^q(t)=  a^q_x(t) + \sum_{x'\in \mathcal S_x}[d^q_{x',x}(t-1)- d^q_{x,x'}(t-1)]$.

\subsection{Interruptible service (IS)}  \label{cat2}
The IS is a simple variation of the previous model, where there is only one possible positive rate of charge $\partial x_i(t)=G_i$ or else,  $\partial x_i(t)=0$. This category best models pool pumps or EVs that can only be charged at a certain charging rate, e.g., 1.1 kW or 3.3 kW. We cluster loads based on $\boldsymbol{\theta}_i = (E_i, \chi_i, \rho_i, G_i)$, and the population plasticity is:
 \begin{align} %
{\mathcal L}^v(t) \!=\! &\Big\{L(t) |  L(t) \!= \!\sum_{q=1}^{Q^v}\sum_{x=0}^{E^q} \!(x'\!-\!x)\partial d^q_{x,x'}(t)|_{x'=\min\{x+G^q,E^q\}}, \nonumber  \\ 
&\partial d_{x,x'}^q(t) \in \mathbb{Z}^+, \partial d_{x,\min\{x+G^q,E^q\}}^q(t)\leq n^q_x(t),\nonumber \\ 
&  n_x^q(\chi^q) = 0, \forall x < \rho^q E^q\Big\} \label{occur6} 
\end{align}
The model can be further complicated by considering a rate $G^q_x$ that is state dependent. This can 
 capture uneven EV load profiles during the initial and final phases of charging \cite{kesidis}.

\subsection{ Non-interruptible Deferrable Service (NID)}\label{cat4}
This category best models appliances such as washer/dryers, and non-interruptible EV charging.
Each appliance $i$ in this category is characterized by $(\tau_i,\chi_i,\ell_i(t))$, where $\tau_i$ is the arrival time of the request, $\chi_i$ is the maximum tolerable delay for the task to start, and $\ell_i(t)$ is a pulse that captures the load profile of appliance $i$ (if it is turned on at $t=0$). 
Hence, the only control available is shifting the load by a delay $f_i$.  The 
plasticity is  \begin{align}
{\mathcal L}_i(t)=\{L_i(t)| L_i(t)=\ell_i(t-f_i),\tau_i \leq f_i \leq t + \chi_i\}
\end{align}
The description above can be replaced with the following integer linear model, based on the state $x_i(t)$ of the ON switch:
\begin{align}
{\mathcal L}_i(t)=&\{L_i(t)| L_i(t)=\ell_i(t)\star \partial x_i(t), x_i(t) \in \{0,1\},\label{eq.defer}\\
	& ~x_i(t)\geq a_i(t-\chi_i), ~x_i(t-1)\leq x_i(t)\leq a_i(t) \}.\nonumber
\end{align}
with $x_i(t) =  0$ and $1$ respectively corresponding to off and on status, and $\star$ denotes the convolution operation.  
Note that $x_i(t)$ can only be in the form of a step function due to the constraints.
The convolution with $\partial x_i(t)=\delta(t-f_i)$ yields
$L_i(t)=\ell_i(t-f_i).$ The constraints ensure causality and that delay does not exceed the maximum tolerable amount.


In this case, appliances are clustered based on a quantized profiles $\ell^{q}(t)$ that most resembles their behavior once on, plus the maximum acceptable delay $\chi^q$. All arrivals that belong to a cluster $a^q(t)$ are in the same same OFF state, and therefore a single process $d^q(t)$ per cluster is needed to account for the activations. 
The population plasticity is:
\begin{align}\label{lint}
{\mathcal L}^v(t) \!=\! \Big\{&L(t) |  L(t) \!= \!\sum_{q=1}^{Q^v} \ell^q(t)\star \partial d^q(t), d^q(t)\in\mathbb{Z}^+ \\ 
& ~d^q(t)\geq a^q(t-\chi^q), ~d^q(t-1)\leq d^q(t)\leq a^q(t)
\Big\} \nonumber
\end{align}
Unlike the previous plasticity models that are closer to discrete time linear systems with integer constrains, the change in state triggers a power injection that lasts for the entire duration of the injection pulse $\ell^q(t)$. This model corresponds to a hybrid systems since it includes discrete switching as well as a dynamics; there are only two hybrid states, one being the system at rest, and the other is the system evolution once ON, captured by $\ell^q(t)$. The pulse $\ell^q(t)$ can be viewed as the impulse response of the dynamical system when it switches to the ON state at time zero. 

The modeling for the NID can be extended to include
more hybrid states $x\in {\mathcal X}^q$ than just a single switch, modeling non-interruptible tasks that follow each other, with a single shared deadline (e.g. a washer and dryer cycle).

\section{Planning, Control, Pricing and Information}\label{sec.why}
As mentioned in the Introduction, hybrid stochastic load models can help planning and real-time control decisions of an Aggregator. For example, in a  two-settlement structure:

{\bf Ex-ante}: the Aggregator plans how much power $B(t)$ to purchase and how much ancillary service  capacity $M(t)$ to offer in the forward market. The Aggregator picks $(B(t),M(t))$  to minimize a cost ${\mathcal C}^{F}(L(t),B(t),M(t))$, which includes costs and benefits of buying energy and selling ancillary services over a certain time horizon $\Omega$, i.e. 
\begin{equation}\label{ex-ante}
\min_{B(t),M(t)} \sum_{t\in \Omega}{\mathbb E}\{{\mathcal C}^{F}(L(t),B(t),M(t))\}~~\mbox{s.t.}~ L(t)\in { \mathcal L}^{DR}(t),
\end{equation}
where ${\mathcal L}^{DR}(t) \subseteq {\mathcal L}(t)$ denotes the set of possible load shapes that can be extracted from a population with load plasticity ${\mathcal L}(t)$ under a specific demand response strategy exercised by the Aggregator. The expected value averages over the appliance arrivals statistics. For brevity, we assume that $(B(t),M(t))$ will be cleared at the marginal price by the market operator.

{\bf Real-time}: the Aggregator  is committed to control $L(t)$ to follow the schedule $(B(t),M(t))$ for the current time $t$ 
and minimize its real-time cost. The Aggregator can be myopic:
\begin{equation}\label{myopic}
\min_{L(t)} ~{\mathcal C}^R(L(t),B(t),M(t))~~\mbox{s.t.}~ L(t)\in { \mathcal L}^{DR}(t),
\end{equation}
or be foresighted, solving a model-predictive problem similar to \eqref{ex-ante}, with the real-time cost ${\mathcal C}^R$ as opposed to ${\mathcal C}^{F}$ for a time horizon $\Omega=\{t,t+1,\ldots,t+H\}$ that includes $H$ dummy future decisions, constantly revised and updated in real-time. 
Note that such model-predictive strategies will be very hard to implement if the Aggregator considers the continuous characteristics of every single individual load. Thus, without clustering, the Aggregator may need to resort to myopic policies for real-time control.

Next we discuss the set ${ \mathcal L}^{DR}(t)$ for the two competing solutions for continuous end-use demand management, i.e., direct load scheduling and dynamic pricing.

\subsection{Incentive-Based Direct Load Scheduling (DLS)}\label{drtypes} For this type of demand management, the Aggregator directly observes the arrivals of the controlled appliances and decides the exact load shape for the aggregate load of the population, choosing from the set ${\mathcal L}^{DR}(t)$. 

This makes DLS the most reliable form of load management possible because the Aggregator has full information, unlike in dynamic pricing.
But what is ${\mathcal L}^{DR}(t)$?  This set strictly depends on the choice of end-use customers to participate in the DLS program and allows the Aggregator to schedule their consumption. A rational customer will not allow an Aggregator to exploit load plasticity for free on a regular basis. Let us denote by $i^v_\vartheta(t)$ the incentives available to appliances in category $v$ with parameters $\vartheta$ participating in the DLS program.  The vector that contains all incentives for $ \vartheta \in \mathcal{T}^v$ in  category $v$ is $\boldsymbol{i}^v(t)$. Then, 
\begin{equation}{\mathcal L}^{DR}(t) = {\mathcal L}^I(t;\boldsymbol{i}^v ) + \sum_{v=1}^V\sum_{\boldsymbol{\vartheta}\in \mathcal{T}^v}n^v_{\boldsymbol{\vartheta}}( \boldsymbol{i}^v){\mathcal L}^v_{\boldsymbol{\vartheta}}(t). \end{equation}
where the number of flexible appliances is modeled as a function of the incentives.
Note that the incentives $\boldsymbol{i}^v(t)$ have two main effects: 1) they could affect the participation rate of customers in the DLS program; 2) they could affect the consumption habits of end-users. For example, if a customer observes that longer EV charge cycles are paid less incentives per unit charge, then they could charge their EV more often to receive more discount. Another example would be that a customers may run the dishwasher more often when they can get lower billing rates due to DLS participation.  These two effects present uncertainty to the Aggregator and need to be modeled for incentive design.

 However, the good news is that, no matter how well the economic side of the problem is handled, the control performance is not directly affected. The number of recruited appliances  $n^v_{\boldsymbol{\vartheta}}(\boldsymbol{i}^v)$ is observable after posting the incentives and before control. Thus, as opposed to the economic aspects of the problem, the Aggregator faces no uncertainty in {\it real-time control} and ${\mathcal L}^{DR}(t)$ is deterministically known.

 The design of appropriate economic incentives for direct load management is the subject of ongoing research, see, e.g., \cite{alizadeh, bitar2012deadline,kefayati}, and will not be discussed in detail.

\subsection{Dynamic Retail Pricing} In this type of demand management, time-varying retail prices constitute the only knob available to the Aggregator for both real-time control and billing. Thus, ${\mathcal L}^{DR}(t)$ might not necessarily be equal to the inherent load plasticity of the population  ${\mathcal L}(t)$. Denoting the vector of retail prices posted for the time horizon $t \in \{1,\ldots,T\}$ as $\mathbf{p}(t) = [\pi^r(1),\ldots, \pi^r(T)]$, this type of DR amounts to the following set of possible load shapes:
\begin{align}
{\mathcal L}^{DR}(t) = \big\{ L(t) |& L(t) = f(t;\mathbf{p}(t)), \mathbf{p}(t) \in \mathcal{Z}(t) \big\}
\end{align}
where $\mathcal{Z}(t)$ is the set of all possible retail prices (probably partially regulated), and the function $f(.)$ denotes the price-response of the population, which is equal to: 
\begin{align}\label{priceresp}
 f(t;\mathbf{p}(t)) = &L^I(t) + \\&\sum_{v=1}^V\sum_{\boldsymbol{\vartheta}\in \mathcal{T}^v} n^v_{\boldsymbol{\vartheta}}(\mathbf{p}(t)) {\displaystyle \argmin_{L(t) \in \mathcal{L}^v_{\boldsymbol{\vartheta}}(t)}} \sum_{t = 1}^{T} \pi(t) L(t). \nonumber
\end{align}

Given that we know the plasticity $\mathcal{L}^v_{\boldsymbol{\vartheta}}(t)$,  the unknown terms in the price response function are the number of appliances in each category $v$ with characteristic vector $\boldsymbol{\vartheta}$, i.e., $n^v_{\boldsymbol{\vartheta}}(\mathbf{p}(t))$. However, as opposed to the DLS case,  these numbers are not observable to the Aggregator after posting the prices and the customer response is harder to learn. Also, the economic side and the control side of the problem are fully tied together here,  which further complicates planning.
\subsection{Information for DLS and Dynamic Pricing}

These are the components necessary for the DLS:
\begin{enumerate}
\item \textit{Uplink  information}: The Aggregator needs the real time values of the arrival processes $a_x^{q,v}(t)$ for each state $x$, cluster $q$ and category of appliance $v$; for a Aggregator with a portfolio of $V$ appliances categories, the expected number of messages per unit time $\delta T$ is:
\begin{equation}
\sum_{v=1}^V\sum_{q=1}^{Q^v}\sum_{x=1}^{E^q} \mathbb{E}\{\partial a_x^{q,v}(t)\}
\end{equation}
\item  \textit{Downlink  information}: The Aggregator needs to control the appliances;  
\item {\it Measurement and verification}: A slower channel is necessary to measure that the control actions are correctly executed, to properly account for the benefits $\boldsymbol{i}^v$ or penalties associated with the appliance response. 
\end{enumerate}
We claim that all the uplink and downlink communication can be designed to keep participating customers anonymous. Retaining anonymity means that the Aggregator can be blind to the appliance owner's identity and still function. To enforce {\it Measurement and verification}  and retain anonymity, it is necessary to enlist a third party for billing. Note that anonymity and privacy are not the same thing, as we discuss later in Sec. \ref{sec.privacy}.
 
In the uplink, knowing $a_x^{q,v}(t)$ does not require customer identity. Thus, third party {\it collectors} can be in charge of gathering arrival information across the grid and forward the values of $\partial a_x^{q,v}(t)$ to the Aggregator. The downlink control can also be handled using a single broadcast message to the entire population, without addressing specific appliances. In fact, the Aggregator decides at time $t$ to move $d_{x,x'}^{q,v}(t)$ appliances from state $x$ to state $x'$  in a certain cluster $q$ of category $v$, and to activate them it can simply broadcast at every $t$ the following table:
\begin{equation}
\boldsymbol{\kappa}_x^{v}(t) = \left\{ \frac{d_{x,x'}^{q,v}(t)}{n_x^{q,v}(t)}, \forall v, \forall q \in \{1,\ldots Q^v\}, \forall x'\in \mathcal S^{q,v}_x\right\}.
\end{equation}
Upon receiving the table, if an appliance of category $v$ happens to be in cluster $q$ at state $x$, they have to move to state $x'$ with probability $ \frac{d_{x,x'}^{q,v}(t)}{n_x^{q,v}(t)}$. Given a sufficiently large population size, this randomized scheme will match closely the proportions scheduled, while keeping the actions of the customer private.

Dynamic pricing, instead, uses the following components:
\begin{enumerate}
\item \textit{Uplink  information} and {\it Measurement}: The Aggregator meters household consumption. 
\item  \textit{Downlink  information}: The Aggregator  posts prices for power consumption during a certain period. 
\end{enumerate}

\subsection{Privacy and Fairness}\label{sec.privacy}
In dynamic pricing, Aggregators get to observe the aggregate load, not individual uses. The consumption of the house includes various loads and the question is how easy it is to separate these contributions and determine all activities in the household. As argued in \cite{lalitha}, this issue can be interpreted as a special instance of {\it data differential privacy}, which is concerned with  information leaks from queries to databases \cite{diff2}. For the observer that looks at smart meter data, the database is a virtual one, with all the individual load injections separated and labeled, while the {\it query} produces their unlabeled sum.  The analysis in \cite{lalitha} explores the 
trade-off between privacy (reconstruction of individual appliance uses) and utility from the side of the observer, measured in terms of reconstruction error. 
In the case of the DLS, this question has not been asked before, but it turns out that differential privacy issues still exist even though customers could be anonymous under our proposed clustering approach. We address this concern next.

How much information leaks about a specific (anonymous) individual from the aggregate query $a_x^{q,v}(t)$?
Here we have  collectors with streaming databases $\boldsymbol{D}^v(t), v=1,\ldots,V$ that include the information about the index, cluster, and state $(i,q_i,x_i)$ of each individual appliance that arrives in category $v$.  At time $t$, every collector answers a specific database query  by the Aggregator on $\boldsymbol{D}^v(t)$ to obtain 
$a_x^{q,v}(t)$. These queries form the Aggregator database $\boldsymbol{A}^v(t), v=1, \ldots, V$. The Aggregator can then attempt to infer information about individual customers that can eventually lead to resolve their identity from $\boldsymbol{A}^v(t)$.

For each $\delta T$, the amount of information that leaks  to the Aggregator can be quantified as the mutual information between
 $\boldsymbol{D}^v(t)$ and $\boldsymbol{A}^v(t)$, denoted by $I(\boldsymbol{A}^v(t);\boldsymbol{D}^v(t))$ \cite{cover}.  
 Attempting to recover  $\boldsymbol{D}^v(t)$ from $\boldsymbol{A}^v(t)$ from a record of values $\boldsymbol{A}^v$ will incur an error no smaller than 
 what can be determined using the bound in \cite{hanverdu} (which improves upon the celebrated Fano's inequality \cite{cover}):
  \begin{equation}
\epsilon_{\boldsymbol{D}^v|\boldsymbol{A}^v}\geq 1 +\frac{I(\boldsymbol{A}^v;\boldsymbol{D}^v)+\log 2}{\log(\max Pr(\boldsymbol{A}^v))}
\end{equation}
In order to evaluate this bound, $\boldsymbol{D}^v(t)$ can be seen as the output of multiple {\it binary noiseless adder channels}. The inputs are Bernoulli random variables with a success probability $\epsilon_{x,i}^{q}(t)$, denoting the probability that a request from appliance $i$ arrives at time $t$, and is assigned to state $x$ of cluster $q$.
It is straightforward \cite{cover} to show that $I(\boldsymbol{A}^v(t);\boldsymbol{D}^v(t))$ is upper-bounded by:
 \begin{equation}
I(\boldsymbol{A}^v(t);\boldsymbol{D}^v(t)) \leq \sum_{q=1}^{Q^v}\sum_{x=0}^{E^q}  H(z^{q,v}_x)
\end{equation}
where $z^{q,v}_x \sim \mbox{Poisson}(\sum_{i \in {\mathcal P}^v} \epsilon_{x,i}^{q}(t))$.
This clarifies that anonymity does not guarantee privacy in DLS either. 

However, 
the trade-off for dynamic pricing and DLS is different. 
Having to design physically reliable prices for control, fair or not, Aggregators using dynamic pricing  need to learn the underlying structure of load and model the price-response  \eqref{priceresp}. 
This creates a greater conundrum in dynamic pricing between privacy and reliability that does not really exists in the anonymous DLS case. In an anonymous DLS scheme, appliance consumptions are separated but customers identities are not needed for the control. Therefore there is a natural separation between 
the information that is needed for billing and the information that is needed for control. This aspect can be advantageous for fair pricing.

\section{Numerical Case Study}\label{sec.numerical}
In our numerical tests we focus on DLS, which represents an upper-bound for what can be attained with dynamic pricing.
To showcase the benefits of quantized population load models, we study an Aggregator that manages the charging load of 40000 Plug-in Hybrid Electric Vehicles (PHEV) on a daily bases under a DLS program. All PHEVs studied here arrive in the grid for non-interruptible  level-1 (1.1 kWs) charging based on arrival and charge amount statistics presented in \cite{alizadeh2013ev}. 
PHEV owners decide whether to participate in the DLS program or not, and how much laxity to offer to the Aggregator, based on incentives $\boldsymbol{i}^v(t)$ designed in \cite{alizadeh}. Thus, here we take $n^v_{\boldsymbol{\vartheta}}(\boldsymbol{i}^v)$ as given.

\begin{figure}
\centering
\includegraphics[width = \linewidth]{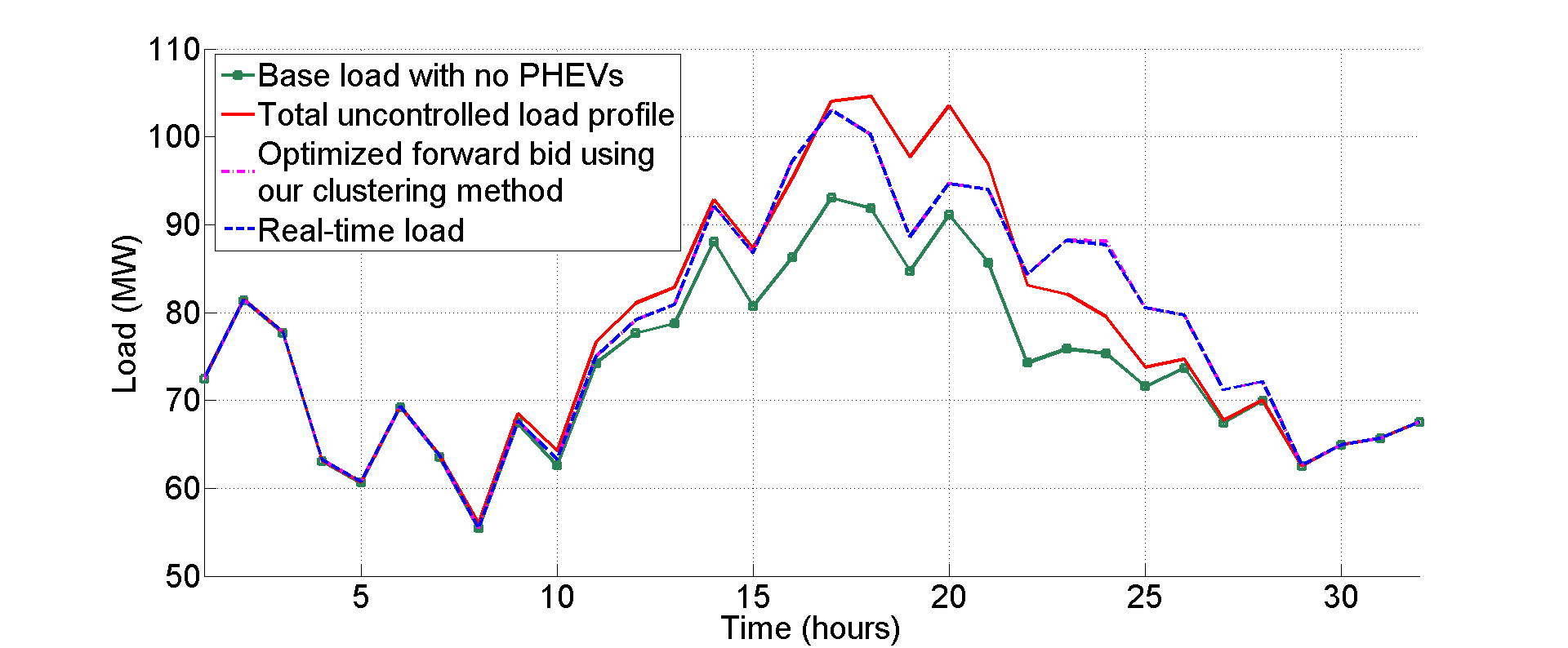} 
\caption{The aggregator can perfectly follow the day-ahead bid that was optimized under our proposed demand clustering method.}
\label{fig:clustering}
\end{figure}

\begin{figure}
\includegraphics[width = \linewidth]{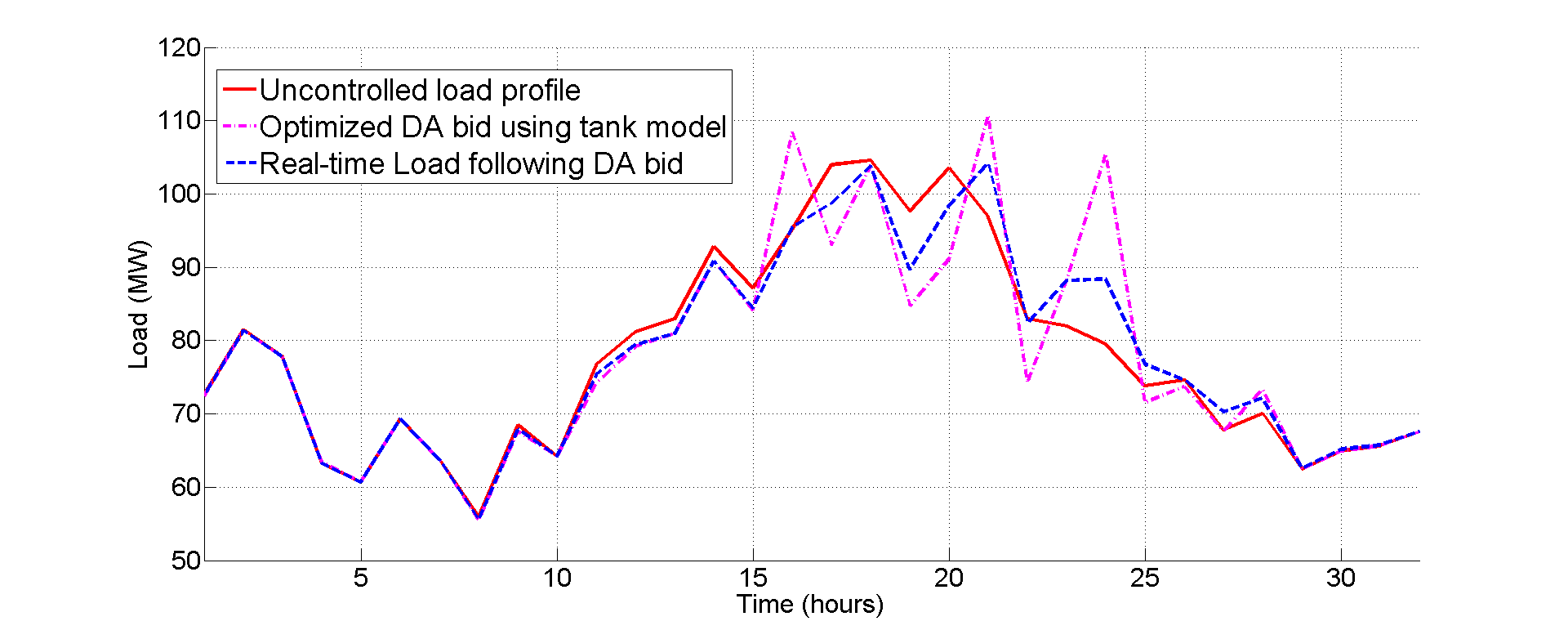} 
\caption{The aggregator cannot follow the day-ahead bid optimized using the tank model, simply due to modeling errors.}
\label{fig:tank}
\end{figure}

Participating PHEVs in the DLS program are clustered into 15 different clusters based on the amount of charge required (quantized with 1 kWh steps into 5 levels) and charge laxity (quantized with 1 hour steps into 3 levels). The Aggregator needs to use the flexibility of recruited PHEVs to minimize its energy market costs. We assume a simple cost structure. The forward market cost is determined by a forward purchase $B(h)$ for each hour $h$ of the day, multiplied by the price $\pi^F(h)$. The real-time cost is based on deviations from $B(h)$, both upward and downward. Market prices reflect day-ahead PJM prices on October 22nd, 2013.

The performance of a quantized population model in determining the optimal  $B(h)$ is compared to that of the tank model  \cite{homer, 6426102,ortega}. A tank model effectively treats every appliance as an canonical battery, thus discarding the information that the PHEV charge here is not interruptible and can only happen at a 1.1 kW rate. Both the tank model and our population model in \eqref{lint} are linear, making the optimization an  integer linear program. Stochasticity of vehicle arrivals is handled through sample average approximation. We extend the length of each day to 32 hours to show the effect of PHEV charges that are deferred into the next day. Figures  \ref{fig:clustering} and \ref{fig:tank} show the performance of each model in generating $B(h)$ that can be followed in real-time by a direct scheduler.
Notice that the overly optimistic $B(h)$ determined by the tank model leads to large real-time costs.

\section{Conclusions}
We presented a scalable model to quantify the load flexibility of a large population of small heterogeneous appliances.
The scheme can be naturally generalized to cover other categories of appliances not studied here, e.g., thermostatically controlled loads. We leave for future work the study of the application of our proposed hybrid models to learn the price response of a heterogeneous population.

\bibliographystyle{IEEEtran}
\small
\bibliography{New2,new4,new5,newrefs}

\end{document}